\documentclass[11pt,letterpaper]{article}

\usepackage[english]{babel}

\usepackage{graphicx} 
\usepackage[margin=1in]{geometry}
\usepackage[english]{babel}

\usepackage{xspace}
\usepackage{booktabs}
\usepackage{caption}
\usepackage{subcaption}
\usepackage{amssymb}
\usepackage{amsmath}
\usepackage{amsthm}
\usepackage{bbm}
\usepackage{graphicx}
\usepackage{multicol}
\usepackage{diagbox}
\usepackage{etoolbox}
\usepackage{pgfplots}
\pgfplotsset{compat=1.18}
\usepackage[colorlinks=true, allcolors=teal]{hyperref}
\usepackage{enumerate}    
\usepackage[ruled,linesnumbered,vlined]{algorithm2e}
\newtheorem{theorem}{Theorem}[section]

\newtheorem{lemma}[theorem]{Lemma}

\usepackage{tikz}
\usepackage{siunitx}
\usepgfplotslibrary{fillbetween}
\usepackage{framed}
\usepackage{mdframed}
\usepackage{xcolor}
\definecolor{shadecolor}{gray}{0.85}
\usepackage{tcolorbox}
\tcbset{
  myframe/.style={
    colback=gray!5,
    colframe=black,
    boxrule=0.8pt,
    arc=4pt,
    outer arc=4pt,
    boxsep=5pt,
    left=5pt,
    right=5pt,
    top=5pt,
    bottom=5pt,
    enhanced,
    sharp corners
  }
}
\DeclareFontShape{OT1}{cmr}{m}{scit}{
    <-> ssub * cmr/m/sc 
}{}

\newmdenv[
  linecolor=gray,
  linewidth=0.9pt,
  leftmargin=0.6em,
  rightmargin=0.6em,
  backgroundcolor=gray!5,
  innertopmargin=0.5em,
  innerbottommargin=0.5em,
  skipabove=0.5em,
  skipbelow=0.5em,
  frametitle={},
]{researchquestion}

\newcommand{\ALG}{\mathsf{ALG}}
\newcommand{\OPT}{\mathsf{OPT}}
\newcommand{\E}{\mathbb{E}}

\renewcommand{\d}{\mathrm{d}}

\title{Minimization Prophet Inequality with Bounded Costs}

\author{
Haolong Li 
\thanks{University of Macau. 
\{yc47435, xiaoweiwu\}@um.edu.mo}
\and 
Xiaowei Wu $^*$
}
\date{}

\begin{document}

\begin{titlepage}
\maketitle
\thispagestyle{empty}

\begin{abstract}
We study the \emph{cost-minimization prophet inequality} problem, in which a decision-maker sequentially observes $n$ independent and identically distributed (IID) random variables. After each observation, the decision-maker must either accept the current realization and stop, or reject it and continue with the next variable. The goal is to minimize the selected value. Unlike the classical maximization setting, if none of the first $n-1$ values is accepted, the final realization must be selected.
Esfandiari et al.\ (SIDMA 2017) showed that this problem does not admit constant-competitive online algorithms in general, which motivates subsequent work on restricted distribution classes, such as entire distributions (Livanos and Mehta, SODA 2024) and distributions with bounded extreme values (Livanos and Mehta, EC 2025).

In this work, we focus on distributions with bounded support. For distributions supported on $[1,b]$, we characterize the competitive ratio of online algorithms as a function of both $b$ and $n$. We consider both distribution-aware and distribution-oblivious settings. In the distribution-aware case, we provide a nearly tight characterization of the optimal online algorithm, showing that its competitive ratio is at most $b^{(1-1/n)^n}$ for all $n\geq 2$.
More significantly, we design a distribution-oblivious algorithm that achieves the same guarantee, with a sequence of predetermined thresholds that depend only on $b$ and $n$.
Furthermore, we extend our analysis to the non-IID setting, where we show that a simple single-threshold algorithm attains an asymptotically optimal competitive ratio of $\Theta(\sqrt{b})$.
\end{abstract}

\end{titlepage}
\section{Introduction}
Prophet inequality is a central topic in optimal stopping theory, initiated by Krengel and Sucheston~(\cite{krengel1977semiamarts, krengel1978semiamarts}). 
In this problem, $n$ random variables $X_1,\ldots, X_n$ reveal their values sequentially.
Upon observing each realization, the algorithm must make an irrevocable accept-or-reject decision. 
If a variable is selected, the process terminates; otherwise, the algorithm proceeds to the next variable. 
The objective is to maximize the selected value. 
The performance of an online algorithm is measured by its competitive ratio, which is defined as the ratio between its expected payoff and the expected maximum realization (the so-called ``prophet''). 
Krengel and Sucheston showed that the optimal competitive ratio is $1/2$, which is achievable via a simple single-threshold algorithm.
Over the past decades, numerous variants of the prophet inequality problem have been studied, including settings with random arrivals~(\cite{siamdm/EsfandiariHLM17, sigecom/CorreaFHOV17, conf/soda/EhsaniHKS18, sigecom/AzarCK18, DBLP:journals/mp/CorreaSZ21, conf/soda/Harb25}), order selection constraints~(\cite{sigecom/AzarCK18, journals/ior/BeyhaghiGLPS21, stoc/ChawlaHMS10, soda/CorreaSZ19, DBLP:conf/focs/PengT22, DBLP:conf/sigecom/BubnaC23}), and independent and identically distributed (IID) variables~(\cite{hill1982comparisons, stoc/AbolhassaniEEHK17, journals/mor/CorreaFHOV21}).

In a complementary direction, researchers have considered the cost-minimization variant, where $X_1,\ldots, X_n$ represent costs and the objective is to minimize the selected value. 
Here, the prophet corresponds to the expected minimum realization. 
Unlike the maximization setting, if none of the first $n-1$ values is accepted, the algorithm must select the final realization, as otherwise rejecting all variables would be trivially optimal. 
Despite extensive progress on the maximization version, where constant competitive ratios are known in many settings, the minimization variant presents fundamentally different challenges due to its upward-closed structure and remains much less understood.
Esfandiari et al.~\cite{esa/EsfandiariHLM15} were among the first to study this problem and established a strong hardness result using the following two-variable instance:
\begin{align*}
    X_1 = 1, \quad 
    X_2 = \begin{cases}
        \frac{1}{\epsilon} & \text{with probability } \epsilon, \\
        0 & \text{with probability } 1-\epsilon.
    \end{cases}
\end{align*}
In this example, the optimal online algorithm incurs an expected cost of $1$ (since $X_1 = \E[X_2] = 1$), whereas the prophet’s expected cost is $\epsilon$, implying that no constant competitive ratio is achievable for general distributions. 
They later extended this construction to IID distributions~\cite{siamdm/EsfandiariHLM17}, showing that even in the IID setting, constant-competitive algorithms do not exist.
Livanos and Mehta~\cite{conf/soda/LivanosM24} further strengthened this hardness result by exhibiting an instance with $n=2$ IID variables supported on $[1,+\infty)$, with cumulative distribution function (CDF)
\begin{equation*}
    F(x) = 1 - \frac{1}{x}.
\end{equation*}
For this distribution, $\E[X] = +\infty$, so the expected cost of any online algorithm is unbounded, whereas the prophet achieves a constant expected cost, ruling out any constant competitive ratio.

Motivated by these strong negative results, prior work has focused on restricted classes of distributions. 
Livanos and Mehta~\cite{conf/soda/LivanosM24} studied \emph{entire} distributions, whose cumulative hazard rate admits a convergent Puiseux series expansion over the support. 
For this class, they showed that the optimal online algorithm achieves a constant-factor approximation, where the constant depends on the distribution’s hazard rate. 
Subsequently, Livanos and Mehta~\cite{conf/sigecom/LivanosM25} employed Extreme Value Theory (EVT) to analyze the problem. 
In the large-market regime as $n \to \infty$, they established that the asymptotic competitive ratio of the optimal online algorithm converges to a closed-form function $\Lambda(\gamma)$, which depends only on the distribution’s extreme value index $\gamma$.

While these frameworks capture a broad range of distributions, they inherently assume unbounded support. 
In contrast, many real-world cost-minimization problems feature a bounded ratio between the maximum and minimum realizations, i.e., distributions with bounded support. 
We illustrate this with the following example:

\begin{tcolorbox}[colback = gray!15, colframe = gray!15]
\textbf{Flight Ticket Booking:}
A researcher needs to purchase a flight ticket to attend a conference. 
Ticket prices fluctuate over time due to seasonal demand, fuel prices, and airline competition, and can thus be modeled as random. 
Before the conference deadline, the researcher can repeatedly check prices (e.g., daily or weekly), but must make an irrevocable take-it-or-leave-it decision at each observation. 
Since attendance is mandatory for accepted papers, a ticket must be purchased before the deadline.
\end{tcolorbox}

In this example, ticket prices can reasonably be modeled as IID, decisions are irrevocable, and the objective is to minimize the purchase price subject to a deadline.
All these features precisely match the cost-minimization prophet inequality framework. 
Importantly, practical considerations such as pricing regulations and market constraints imply that prices are typically bounded.
For instance, it is unlikely to observe prices below \$300 or above \$3{,}000 under normal conditions.

This motivates the study of distributions with bounded support. 
By normalization, we may assume the support is $[1,b]$, where $b>1$ captures the ratio between the largest and smallest possible realizations.
We refer to such distributions as \emph{$b$-bounded}. 
For $b$-bounded distributions, any online algorithm trivially has a competitive ratio at most $b$, and it is natural to ask whether significantly better guarantees are achievable. 
This leads to the following central question:
\begin{center}
\emph{What is the optimal competitive ratio, as a function of $b$, for cost-minimization prophet \\inequality with $b$-bounded distributions?}
\end{center}

Notably, truncating the two-variable IID hard instance of Livanos and Mehta~\cite{conf/soda/LivanosM24} with CDF $F(x)=1-1/x$ to the range $[1,b]$ yields an instance for which the optimal online algorithm has a competitive ratio of $\Omega(\ln\ln b)$. 
It remains open whether an $O(\ln\ln b)$-competitive algorithm exists, or whether stronger lower bounds hold.

\smallskip

Furthermore, in many practical settings, while the range of possible costs may be known, the underlying distribution is not. 
For example, airline prices depend on numerous unpredictable factors, including fuel costs, demand fluctuations, and geopolitical conditions, which makes it unrealistic to assume that the underlying distribution is known. 
This motivates a second fundamental question:
\begin{center}
\emph{Can we design competitive online algorithms for unknown $b$-bounded distributions?}
\end{center}

In this paper, we address both of these questions by studying cost-minimization prophet inequalities under known and unknown $b$-bounded distributions.

\subsection{Our Contributions}

Our first result addresses the open question regarding the competitive ratio of the optimal online algorithm for known $b$-bounded distributions. We provide a nearly tight characterization (up to polylogarithmic factors) of this ratio.

\begin{theorem}
\label{thm:informal-online-optimum}
For the cost-minimization IID prophet inequality with a known $b$-bounded distribution, the optimal online algorithm achieves a competitive ratio of $b^{(1-1/n)^n}$. Moreover, there exists an instance for which its ratio is at least $\frac{b^{(1-1/n)^n}}{1+\ln b}$.
\end{theorem}

Our result almost settles the competitive ratio of the optimal online algorithm for $b$-bounded distributions. Our hard instance generalizes and strengthens that of~\cite{conf/soda/LivanosM24}, and may be of independent interest.

\medskip

We then turn to distribution-oblivious algorithms, which know only the support $[1,b]$ but not the exact CDF of the distribution. Without knowledge of the CDF, the optimal online algorithm, which is defined via dynamic programming (see Section~\ref{sec:online-optimum}), is no longer well-defined. Instead, the distribution-oblivious algorithms are determined solely by the parameters $b$ and $n$.
A natural class of such algorithms specifies thresholds $\theta_1,\theta_2,\ldots,\theta_n$ (as functions of $b$ and $n$) to decide whether to accept variables $X_1,X_2,\ldots,X_n$, respectively. These thresholds should be non-decreasing, which reflects the increased willingness to accept higher costs as fewer variables remain. Moreover, to ensure that at least one variable is selected, it is necessary that $\theta_n = b$. The key question is how to set the first $n-1$ thresholds.
In this work, we study both single-threshold algorithms\footnote{Due to the nature of the problem, a single-threshold algorithm for minimization effectively uses two thresholds, as the last variable is accepted with threshold $b$.} (where $\theta_1 = \cdots = \theta_{n-1} = \theta$) and multi-threshold algorithms. For both classes, we establish nearly optimal competitive ratios. Notably, the competitive ratio achieved by our multi-threshold distribution-oblivious algorithm matches that of the optimal online algorithm.

\begin{theorem}
\label{thm:informal-single-threshold}
For the cost-minimization IID prophet inequality with an unknown $b$-bounded distribution, there exists a single-threshold algorithm with competitive ratio $O(b^{(n-1)/(2n)})$, which is asymptotically optimal.
\end{theorem}

\begin{theorem}
\label{thm:informal-multi-threshold}
For the cost-minimization IID prophet inequality with an unknown $b$-bounded distribution, there exists a multi-threshold algorithm achieving a competitive ratio of $b^{(1-1/n)^n}$.
\end{theorem}

We illustrate the exponent of the competitive ratios for the two algorithms in Figure~\ref{fig:exponent-bounds}.

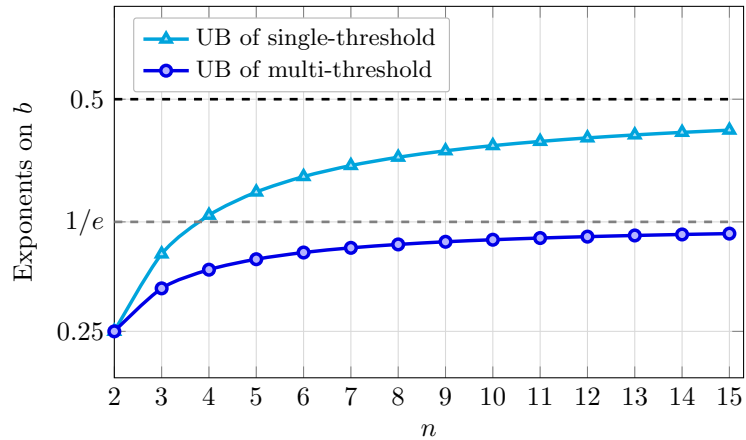
\begin{figure}[ht]
\centering
\begin{tikzpicture}
\begin{axis}[
    width=0.6\linewidth,
    height=6.5cm,
    xlabel={$n$},
    ylabel={Exponents on $b$},
    xmin=2, xmax=15.3,
    ymin=0.2, ymax=0.6,
    xtick={2,3,4,5,6,7,8,9,10,11,12,13,14,15},
    ytick={0.25,0.5,0.75,1.0, 0.3679},
    yticklabels={0.25,0.5,0.75,1.0,$1/e$},
    grid=both,
    grid style={line width=0.2pt, draw=gray!15},
    major grid style={line width=0.4pt, draw=gray!30},
    legend pos=north west,
    legend cell align={left},
    legend style={font=\footnotesize, draw=gray!60, fill=white, rounded corners=1pt},
    samples at={2,3,...,15},
    domain=2:15,
    smooth,
    tick label style={font=\small},
    label style={font=\small},
]

\addplot+[mark=triangle*, mark options={scale=1.1, fill=cyan!30}, color=cyan!90!black, line width=1.3pt]
    { (x-1)/(2*x) };
\addlegendentry{UB of single-threshold}

\addplot+[mark=*, mark options={scale=1.0, fill=blue!30}, color=blue!90!black, line width=1.3pt]
    { pow(1-1/x, x) };
\addlegendentry{UB of multi-threshold}

\addplot+[no markers, line width=1pt, black, dashed, forget plot]
    { 1/2 };
\addplot+[no markers, line width=1pt, gray, dashed, forget plot]
    {1/exp(1)};

\end{axis}
\end{tikzpicture}
\caption{The blue line shows the upper bound on the exponents $(1-1/n)^n$ on $b$ for the multi-threshold algorithm; the cyan line shows that of the single-threshold algorithm, which is $(n-1)/(2n)$.}
\label{fig:exponent-bounds}
\end{figure}

To the best of our knowledge, our work is the first to study the minimization prophet inequality with unknown distributions. In particular, we partially address an open question posed in \cite{conf/soda/LivanosM24}: \emph{if one has only sample access to $D$, how does the competitive ratio of the optimal algorithm depend on the number of samples?} Our results show that, when the distribution has bounded support, a near-optimal online algorithm exists even without any samples.
Moreover, since our algorithms do not rely on distributional information, their design (e.g., the choice of thresholds) depends only on $b$ and $n$, which makes them computationally efficient. We further demonstrate that our approach extends to settings where only $k$ distinct thresholds are allowed (see Appendix~\ref{sec:k-threshold}), thereby partially addressing another open question of \cite{conf/soda/LivanosM24}: \emph{what if we are allowed to use at most $k$ thresholds for $k > 1$?}
We also emphasize that, unlike much of the existing literature, which focuses on asymptotic competitive ratios\footnote{As noted by existing works~\cite{conf/sigecom/LivanosM25,kong2026multiunit}, support-bounded distributions correspond to non-positive extreme value index, for which case the asymptotic competitive ratio is $1$.} (where the distribution is independent of $n$, and they consider the case of $n \to \infty$), all our results hold for every $n \ge 2$ and admit closed-form expressions. For instance, when $n=2$, both our single- and multi-threshold algorithms achieve a competitive ratio of $b^{1/4}$. Although the ratios are polynomial in $b$, the exponents are small, which leads to modest numerical values. As we show in Section~\ref{ssec:numerical}, when $b \le 10$ (a practically reasonable regime), the competitive ratio of our multi-threshold algorithm never exceeds $2.34$ for any $n$. This is particularly relevant given our motivation from practical cost-minimization problems.

\medskip

Our final contribution concerns the non-IID setting, where the $n$ variables may follow different distributions. In this case, we adapt the two-variable hard instance of Esfandiari et al.~\cite{esa/EsfandiariHLM15} to obtain a $b$-bounded instance with $b = \Theta(1/\epsilon^2)$:
\begin{align*}
X_1 = 1, \quad \text{and} \quad
X_2 = \begin{cases}
\frac{1}{\epsilon} & \text{with probability } \epsilon, \\
\epsilon & \text{with probability } 1-\epsilon.
\end{cases}
\end{align*}

For this instance, the optimal online algorithm has a competitive ratio of $\Theta(1/\epsilon)$, which implies an $\Omega(\sqrt{b})$ lower bound for the non-IID setting. A similar lower bound can be shown for the prophet secretary setting. This naturally raises the question: can this bound be achieved?
Our final result answers this question in the affirmative by presenting a distribution-oblivious single-threshold algorithm that is $O(\sqrt{b})$-competitive for all non-IID instances.

\begin{theorem}
\label{thm:informal-hardness-classic}
For the cost-minimization prophet inequality with general (non-IID) $b$-bounded distributions, there exists a single-threshold algorithm achieving a competitive ratio of $2\sqrt{b}$.
\end{theorem}

This result reveals a separation between the best achievable competitive ratios in the IID and non-IID settings for cost-minimization, even when the arrival order is random.

\subsection{A Technical Overview}
As also noted by Livanos and Mehta~\cite{conf/soda/LivanosM24}, a central difficulty in upper bounding the competitive ratio in the cost-minimization setting is bridging the gap between the expressions for the prophet’s cost and that of the optimal online algorithm. Consider the case $n=2$ as an example. The expected cost of the prophet is
\begin{align*}
    \OPT = \int_{0}^{b} G(x)^2 \,\d x,
\end{align*}
where $G(x) = 1 - F(x)$ denotes the survival function.  
The optimal online algorithm accepts the first variable if its realization is at most $\E[X]$; otherwise, it rejects the first variable and accepts the second. Its expected cost is therefore
\begin{align*}
    \ALG = \E[\min\{X,\E[X]\}] = \int_{0}^{\E[X]} G(x) \,\d x.
\end{align*}

To upper bound the competitive ratio, we need to relate these two expressions.\footnote{Existing works~\cite{conf/soda/LivanosM24,conf/sigecom/LivanosM25} used the hazard rates and extreme value theory to bridge this gap.}
We observe that the Cauchy--Schwarz inequality is particularly effective for this purpose. By extending the integral $\int_{0}^{\E[X]} G(x)\,\mathrm{d}x$ as an integral over $[0,b]$ with the integrand set to zero for $x > \E[X]$ and applying the (continuous) Cauchy--Schwarz inequality, we obtain
\begin{align*}
    \int_{0}^{\E[X]} G(x)\,\d x 
    = \int_{0}^{b} G(x)\,\mathbf{1}(x \le \E[X])\,\d x
    \le \left( \int_{0}^{b} G(x)^2\,\mathrm{d}x \right)^{1/2} \cdot \E[X]^{1/2}.
\end{align*}

This bound directly connects $\ALG$ to $\OPT$, which makes it easier to bound the competitive ratio. In particular, applying Cauchy--Schwarz again on $\E[X] = \int_{0}^{b} G(x) \,\d x$, we can derive that $\E[X] \le \sqrt{b}\cdot \OPT$, which implies
\[
\frac{\ALG}{\OPT} \le \left(\frac{\E[X]}{\OPT}\right)^{1/2} = O\!\left(b^{1/4}\right).
\]

Extending the analysis from $n=2$ to general $n$ requires more powerful tools. Fortunately, Cauchy--Schwarz is a special case of Hölder’s inequality, which is sufficiently general to transform the bound on $\ALG$ into a form comparable to $\OPT$. 
For $n \ge 2$, the prophet’s expected cost is $\int_{0}^{b} G(x)^n\,\mathrm{d}x$, while the optimal online algorithm admits the recursive characterization
\[
V_{k+1} = \int_{0}^{V_k} G(x)\,\mathrm{d}x,
\quad \text{with } V_0 = b,
\]
where $V_n$ is the expected cost of the algorithm. Applying Hölder’s inequality gives
\begin{align*}
    V_{k+1} 
    \le \left(\int_{0}^{b} G(x)^n\,\mathrm{d}x \right)^{\frac{1}{n}} \cdot V_k^{\frac{n-1}{n}}.
\end{align*}
Then by an induction on $k$, we can show an upper bound of $b^{(1-1/n)^n}$ on the competitive ratio (see Section~\ref{sec:online-optimum} for the proof).

\medskip

Moving from known to unknown distributions introduces several additional challenges. A primary difficulty is that all of the above expressions depend explicitly on $G(x)$, the (unknown) survival function. Indeed, the optimal online algorithm is no longer well-defined without access to $G(x)$.
Instead, a natural distribution-oblivious approach specifies an increasing sequence of thresholds $\theta_1,\ldots,\theta_n$, depending only on $b$ and $n$, which determine whether to accept each variable.\footnote{Our distribution-oblivious algorithms belong to the class of \emph{blind strategies}~\cite{DBLP:journals/mp/CorreaSZ21}, whose thresholds are fixed without seeing any realizations. In general, one could consider adaptive algorithms that use different thresholds depending on the observed realizations.}
To optimize multi-threshold algorithms, one must identify effective threshold sequences and determine their functional form. 
Consider the case when $n \to \infty$, for example (more details can be found in Appendix~\ref{sec:continuous-analysis}).
After a standard normalization, suppose variables arrive over the interval $[0,1]$, and let $\theta(t)$ denote the threshold applied at time $t$. As discussed earlier, $\theta(t)$ should be increasing and satisfy $\theta(1) = b$. The challenge is to determine its functional form.
Interestingly, both our upper bound analysis and the construction of hard instances suggest that the optimal thresholds follow a doubly exponential form, namely $\theta(t) = b^{e^{\,t-1}}$ (see Figure~\ref{fig:thresholds} for an illustration when $b=10$). Surprisingly, we show that a discretized version of these distribution-independent thresholds achieves the same competitive ratio $b^{(1-1/n)^n}$ as the optimal online algorithm.

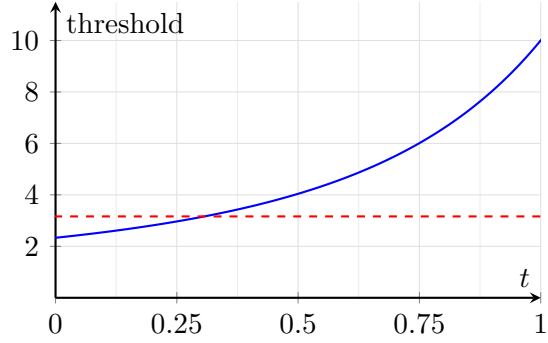
\begin{figure}[htbp]
    \centering
\begin{tikzpicture}
\begin{axis}[
    width=8cm,
    height=5.5cm,
    xmin=0.0, xmax=1.0,
    ymin=0, ymax=11.5,
    xtick={0,0.25,0.5,0.75,1},
    minor x tick num=1,
    ytick={2,4,6,8,10},
    axis lines=middle,
    hide obscured x ticks=false,
    xlabel={$t$},
    ylabel={$\text{threshold}$},
    samples=200,
    smooth,
    thick,
    legend style={at={(0.03,0.97)}, anchor=north west},
    grid=both,
    major grid style={line width=0.3pt, draw=gray!25},
    minor grid style={line width=0.15pt, draw=gray!15},
    minor tick length=0pt,
]

\addplot[blue, thick, domain=-0.2:1.2] {pow(10, exp(x-1))};

\addplot[red, dashed, thick, domain=-0.2:1.2] {sqrt(10)};

\end{axis}
\end{tikzpicture}
\caption{An illustrating example with $b=10$, where the blue line shows the threshold function $\theta(t)=b^{e^{t-1}}$ for the multi-threshold algorithm, and the red dashed line shows the threshold function $\theta(t) = \sqrt{b}$ for the single-threshold algorithm.}
\label{fig:thresholds}
\end{figure}

With these fixed thresholds, the online algorithm becomes well-defined, and we can derive an upper bound on its expected cost analogous to that of the optimal online algorithm. Although the survival function $G(x)$ still appears in the expressions, it serves purely as an analytical tool. 
Using a similar argument, we apply Hölder’s inequality to obtain an upper bound on $\ALG$ that closely matches the lower bound on $\OPT$. Finally, by exploiting the specific form of (the discretized version of) $\theta(t)$, we can eliminate the dependence on $G(x)$ and derive an explicit upper bound on the competitive ratio (see Section~\ref{sec:oblivious-distributions}).
By a similar analysis, we also show that setting $\theta(t) = \sqrt{b}$ yields an asymptotically optimal algorithm that is $O(\sqrt{b})$-competitive.

\subsection{Related Works}
\label{sec:related-work}

\paragraph{Cost-Minimization Prophet Inequality.} 
Livanos and Mehta~\cite{conf/soda/LivanosM24} were among the first to design competitive algorithms for the cost-minimization IID prophet inequality. They studied the class of entire distributions and provided a closed-form expression for the competitive ratio of the optimal online algorithm. They also analyzed single-threshold algorithms and showed that no constant-factor guarantee is possible; instead, the optimal competitive ratio scales as $O(\text{polylog}(n))$, with the exponent depending on the distribution.  
Subsequently, Livanos and Mehta~\cite{conf/sigecom/LivanosM25} revisited the problem using tools from extreme value theory, and obtained upper bounds that depend on the extreme value index of the distributions. Notably, their framework applies to both maximization and minimization settings. Other variants of the cost-minimization prophet inequality have also been studied, including models with hiring time~\cite{journals/mor/DisserFGGKSST20}, multi-choice~\cite{DBLP:conf/esa/AzarBF26} and fractional selection~\cite{journals/mor/QinVW24}.

\paragraph{Maximization Prophet Inequality for Known Distributions.} 
For IID random variables, Kertz~\cite{kertz1986stop} established an upper bound of $0.7451$ on the competitive ratio of the optimal online algorithm. A lower bound of $1 - 1/e$ was shown by Hill and Kertz~\cite{hill1982comparisons}, which was later improved to $0.738$ by Abolhassani et al.~\cite{stoc/AbolhassaniEEHK17}, and eventually matched by the tight $0.7451$ bound of Correa et al.~\cite{sigecom/CorreaFHOV17}.  
For non-IID variables with adversarial arrival order, Krengel and Sucheston~\cite{krengel1977semiamarts,krengel1978semiamarts} established a $1/2$-competitive bound and proved its optimality. Subsequent work has focused on relaxed models, such as the prophet secretary problem (random arrival order) and order selection (where the decision-maker can decide the arrival order). For prophet secretary, competitive algorithms have been developed in~\cite{siamdm/EsfandiariHLM17,sigecom/CorreaFHOV17,DBLP:journals/siamcomp/EhsaniHKS24,sigecom/AzarCK18,DBLP:journals/mp/CorreaSZ21,conf/soda/Harb25}, with the current best lower bound of $0.688$ due to Chen et al.~\cite{DBLP:conf/soda/Chen0LT25} and upper bound of $0.7235$ by Giambartolomei et al.~\cite{DBLP:journals/corr/abs-2304-04024}. For the order selection variant, the best known competitive ratio is $0.7258$ by Bubna and Chiplunkar~\cite{DBLP:conf/sigecom/BubnaC23}, improving upon the $0.7251$ result of Peng and Tang~\cite{DBLP:conf/focs/PengT22}. Beyond multiplicative guarantees, Hill and Kertz
\cite{hill1981additive, hill1982comparisons} derived additive comparisons for bounded independent and IID random variables. Related additive-regret objectives were recently studied by Mahdian et al.~\cite{mahdian2026additively} in the secretary
setting.

\paragraph{Maximization Prophet Inequality for Unknown Distributions.} 
When the distributions are unknown, most work focuses on algorithms with limited sample access. Correa et al.~\cite{DBLP:journals/mor/CorreaDFS22} showed that with $o(n)$ samples, the best achievable competitive ratio for IID prophet inequality is $1/e$, while $O(n^2)$ samples suffice to approach the optimal ratio $0.7451 - \varepsilon$. These bounds were later improved by Rubinstein et al.~\cite{innovations/RubinsteinWW20}, who showed that $O(n)$ samples suffice in the IID setting. They also proved that in the non-IID setting, a single sample from each distribution is enough to achieve the optimal $1/2$ ratio.  
Further progress has been made with additional samples~\cite{DBLP:conf/soda/Ezra26,DBLP:conf/stoc/CristiZ24} and under various sampling models~\cite{DBLP:journals/corr/abs-2104-02050,DBLP:conf/soda/CaramanisDFFLLP22,DBLP:conf/wine/GravinLT22,DBLP:journals/jmlr/CorreaC0S22,DBLP:conf/soda/KaplanNR22,correa2024sample,DBLP:conf/sigecom/0001LTWW024}. There is also a growing line of work studying query access to the CDF. For example, Feng et al.~\cite{DBLP:journals/ai/FengLLWW25} showed that a single query suffices to achieve a ratio of $0.686$, while Perez-Salazar et al.~\cite{DBLP:journals/corr/abs-2210-05634} showed that with two queries, one can achieve a ratio of $0.708$.

\subsection{Organization}

This paper is organized as follows. In Section~\ref{sec:preliminaries}, we introduce the basic models, definitions, and several useful lemmas. In Section~\ref{sec:online-optimum}, we analyze the optimal online algorithm, providing both upper bounds and the matching hard instance. In Section~\ref{sec:oblivious-distributions}, we present the distribution-oblivious algorithms (single- and multi-threshold) and establish their competitive ratios. Finally, in Section~\ref{sec:non-IID}, we study the non-IID setting under both adversarial and random arrival orders.

\section{Preliminaries}
\label{sec:preliminaries}

In this section, we introduce the formal model and relevant notions. We also present several useful lemmas. Since the main focus of this paper is the IID setting, we restrict our notation to this case. The notation for the non-IID setting will be introduced in Section~\ref{sec:non-IID}.

\paragraph{Model.} In the cost-minimization IID prophet inequality model, there are $n$ independent random variables $X_1,\ldots,X_n$ drawn from a common distribution $D$. At each time step, a variable $X_i$ realizes its value, and the algorithm must irrevocably decide whether to accept it as the cost or reject it. The objective is to minimize the cost of the selected realization, subject to the constraint that exactly one variable must be chosen. In this paper, we assume that each variable is supported on $[1, b]$, where $b$ is known; we refer to such distributions as $b$-bounded.

\paragraph{Notations.} Let $F(x)$ denote the cumulative distribution function (CDF) of $D$, and let $G(x) = 1 - F(x)$ denote the survival function. We use $X^* = \min\{X_1, \ldots, X_n\}$ to denote the minimum realization, and let $X$ denote a generic random variable drawn from $D$. We write $\ALG$ for the expected cost of the algorithm and $\OPT = \E[X^*]$ for the expected cost of the offline optimum. The competitive ratio of an algorithm is defined as the supremum of $\ALG/\OPT$ over all instances.

\medskip

We formally state the H\"{o}lder's inequality as follows.

\begin{lemma}[H\"{o}lder's Inequality\footnote{Here we only present a restricted version of the H\"{o}lder's Inequality adapted for our purpose.}]
    If $S \subset \mathbb{R}$ is an interval and $f$ and $g$ are non-negative function such that $f^p$ and $g^q$ are Riemann integrable on $S$, H\"{o}lder's Inequality states that
    \begin{align*}
        \int_{S}  f(x)\cdot g(x) \,\d x \leq \left(\int_S  f(x)  ^p \,\d x\right)^{\frac{1}{p}} \cdot \left( \int_S g(x) ^{q} \,\d x\right)^{\frac{1}{q}},
    \end{align*}
    where $p, q > 1$ are real-numbers satisfying $\frac{1}{p} + \frac{1}{q}=1$.
\end{lemma}

Next, we introduce several lemmas as technical tools.

\begin{lemma}
\label{lem:min-x,t}
    For any $t \in [1,b]$ and random variable $X$ distributed on $[1, b]$ whose survival function is $G(x)$, we have $\E[\min\{X, t\}] = \int_{0}^{t} G(s) \,\d s$.
\end{lemma}
\begin{proof}
    Following the definition of expectation, we have
    \begin{align*}
        \E[\min\{X, t\}] &= \int_{0}^{b} \Pr [\min\{X, t\} > s] \,\d s \\
        &= \int_{0}^{t} \Pr [\min\{X, t\} > s] \,\d s + \int_{t}^{b} \Pr [\min\{X, t\} > s] \,\d s \\
        &= \int_{0}^{t} \Pr [\min\{X, t\} > s] \,\d s
        = \int_{0}^{t} G(s) \,\d s. \qedhere
    \end{align*}
\end{proof}

\begin{lemma}
\label{lem:conditioned-expectation}
For any $\theta \in [1, b]$ and random variable $X$ distributed on $[1,b]$ whose survival functions and CDF are $G(x)$ and $F(x)$, respectively, we have
\begin{equation*}
    \int_{0}^{\theta} G(x) \,\d x = F(\theta) \cdot \E[X\mid X\leq \theta] + \theta \cdot G(\theta).
\end{equation*}
\end{lemma}
\begin{proof}
Following the definition of expectation, we have
\begin{align*}
    F(\theta)\cdot \E[X\mid X\leq \theta] &= F(\theta) \cdot \int_{0}^{b} \Pr[X>x\mid X\leq \theta] \,\d x \\
    &= F(\theta) \cdot \left(\int_{0}^{\theta} \Pr[X>x \mid X\leq \theta]\,\d x + \int_{\theta}^{b}0\ \,\d x \right)\\
    &= F(\theta) \cdot \int_{0}^{\theta} \frac{G(x) - G(\theta)}{F(\theta)}\,\d x
    = \int_{0}^{\theta} G(x) \,\d x - \theta\cdot G(\theta).
\end{align*}    
Rearranging the equality completes the proof.
\end{proof}

\begin{lemma} \label{lemma:upper-bound-G-by-OPT}
    For any $\theta \in [1, b]$ and random variable $X$ distributed on $[1,b]$ whose survival functions is $G(x)$, we have $G(\theta) \leq \left( {\OPT}/{\theta} \right)^{1/n}$.
\end{lemma}
\begin{proof}
    By definition of $\OPT$, we have
    \begin{equation*}
        \OPT = \int_{0}^{b} G(x)^n \,\d x \geq \int_{0}^{\theta} G(x)^n \,\d x \geq \theta\cdot G(\theta)^n,
    \end{equation*}
    where the last inequality holds since $G(\cdot)$ is non-increasing.
    Rearrangement yields the lemma.
\end{proof}

\begin{lemma} \label{lemma:upper-bound-G-integral}
    For any $\theta \in [1, b]$ and random variable $X$ distributed on $[1,b]$ whose survival functions is $G(x)$, we have $\int_{0}^{\theta} G(s) \,\d s \leq \OPT^{\frac{1}{n}} \cdot \theta^{\frac{n-1}{n}}$.
\end{lemma}
\begin{proof}
    By H\"{o}lder's Inequality and the definition of $\OPT$, we have
    \begin{align*}
    \int_{0}^{\theta} G(s) \,\d s &= \int_{0}^{b} G(s) \cdot \mathbf{1} (s \leq \theta) \,\d s \\
    &\leq \left(\int_{0}^{b} G(s)^n\right)^{\frac{1}{n}} \cdot\left( \int_{0}^{b} \left(\mathbf{1}(s\leq \theta)\right)^{\frac{n}{n-1}} \,\d s\right)^{\frac{n-1}{n}} \\
    &= \OPT^{\frac{1}{n}} \cdot \theta^{\frac{n-1}{n}}. \qedhere
    \end{align*}
\end{proof}

\section{Optimal Online Algorithm with Known Distribution}
\label{sec:online-optimum}

For known distributions, the optimal online algorithm can be characterized via a dynamic programming approach. For instance, since the last variable $X_n$ must be accepted if no earlier variable has been selected, the optimal strategy for $X_{n-1}$ is to accept it whenever its realization is at most $\E[X]$. Likewise, the threshold for accepting $X_{n-2}$ is given by the expected cost of the optimal policy applied to the remaining two variables.
Based on this recursive structure, we can precisely describe the behavior of the optimal online algorithm as follows.

\begin{tcolorbox}[title=Online Optimal Algorithm, label=algorithm:online-optimum]
    \textbf{Set $V_0 = b$. Compute thresholds $V_1,V_2,\ldots,V_{n-1}$ as}
    \begin{itemize}
        \item
        for $i = 0,\ldots, n-2$: let $V_{i+1} = \E[\min \{V_i, X\}]$.
    \end{itemize}
    \textbf{Upon the arrival of $X_i$, where $i=1,\ldots,n$:}
    \begin{itemize}
        \item if $X_i \leq V_{n-i}$, accept $X_i$ and terminate.
    \end{itemize}
\end{tcolorbox}

\smallskip

In the following, we give an upper bound for the competitive ratio of the optimal online algorithm and provide a hard instance that nearly matches the upper bound.

\subsection{Upper Bound on the Competitive Ratio}
\label{ssec:ub-online-optimal}

We first characterize the cost of the optimal online algorithm and show the upper bound.

\begin{theorem}
    \label{thm:online-optimum}
    For cost-minimization IID prophet inequality with a known $b$-bounded distribution, the optimal online algorithm achieves a competitive ratio of $b^{(1-1/n)^n}$.
\end{theorem}
\begin{proof}
We proceed in three steps.
First, we derive an explicit expression for the offline optimum $\OPT=\E[X^*]=\E[\min\{X_1, \ldots, X_n\}]$.
Second, leveraging the recursive form of the thresholds, we obtain an upper bound on the expected cost of the optimal online algorithm using Hölder’s inequality.
Third, we use mathematical induction to bound the competitive ratio.

For the expected cost of the prophet, we have
\begin{align*}
    \OPT = \int_{0}^{b} \Pr[X^* > x] \,\d x = \int_{0}^{b} G(x)^n\,\d x.
\end{align*}

For the expected cost of the optimal online algorithm, recall that we use $V_k$ to denote the online optimum with $k$ variables (we set $V_0 = b$ for notational convenience), and we have $\ALG = V_n$.
By the recursive expression of $V_{k}$, we obtain
\begin{align}
\label{eq:alg}
    V_{k} = \E[\min\{V_{k-1}, X\}] = \int_{0}^{V_{k-1}} G(x) \,\d x \leq \OPT^{\frac{1}{n}} \cdot V_{k-1}^{\frac{n-1}{n}},
\end{align}
where the second equality holds by Lemma~\ref{lem:min-x,t} and the inequality holds by Lemma~\ref{lemma:upper-bound-G-integral}.
Now we have a recursion between $V_k$ and $\OPT$. 
In the following, we use mathematical induction to prove that for all $k\geq 0$, we have $V_k \leq \OPT^{1-(\frac{n-1}{n})^{k}} \cdot b^{(\frac{n-1}{n})^{k}}$.
\begin{itemize}
    \item 
    Base case: for $k=0$, we have $V_0=b \leq \OPT^0 \cdot b^1$.
    \item 
    Assume that $V_{k-1} \leq \OPT^{1-(\frac{n-1}{n})^{k-1}} \cdot b^{(\frac{n-1}{n})^{k-1}}$ holds. By Inequality~\eqref{eq:alg}, we have
    \begin{align*}
        V_{k} &\leq \OPT^{\frac{1}{n}} \cdot V_{k-1}^{\frac{n-1}{n}} \\
        &\leq \OPT^{\frac{1}{n}} \cdot \left(\OPT^{1-(\frac{n-1}{n})^{k-1}} \cdot b^{(\frac{n-1}{n})^{k-1}}\right)^{\frac{n-1}{n}} \\
        &= \OPT^{1-(\frac{n-1}{n})^k} \cdot b^{(\frac{n-1}{n})^k},
    \end{align*}
    which completes the induction.
\end{itemize}

This yields an upper bound on the competitive ratio:
\begin{equation*}
   \frac{\ALG}{\OPT} =\frac{V_n}{ \OPT} \leq \frac{\OPT^{1-(\frac{n-1}{n})^{n}} \cdot b^{(\frac{n-1}{n})^{n}}}{\OPT} \leq b^{(1-\frac{1}{n})^n}.\qedhere
\end{equation*}
\end{proof}

\subsection{Lower Bound on the Competitive Ratio}
\label{ssec:lb-online-optimal}

Next, we show that our analysis for the optimal online algorithm is nearly tight by providing a hard instance for which the competitive ratio of the optimal online algorithm is at least $\frac{b^{(1-1/n)^n}}{1+\ln b}$.

\begin{theorem}
    \label{thm:hardness-result-n=infty}
    For cost-minimization IID prophet inequality with a known $b$-bounded distribution, the competitive ratio of the optimal online algorithm is at least $\frac{b^{(1-1/n)^n}}{1+\ln b}$.
\end{theorem}
\begin{proof}
Consider a distribution whose survival function is 
\begin{align*}
    G(x) = 
    \begin{cases}
        x^{-\frac{1}{n}} \quad & \text{when } 1\leq x < b,\\
        0 & \text{when } x\geq b.
    \end{cases}
\end{align*}

This hard instance is inspired by that of \cite{conf/soda/LivanosM24} (whose CDF is $F(x) = 1-1/x$, as mentioned in the introduction).
By the definition of $\OPT$, we have
\begin{align*}
    \OPT &= \int_{0}^{b} G(x)^n \,\d x 
    = 1 + \int_{1}^{b} x^{-1} \,\d x
    = 1+ \ln b.
\end{align*}

We still use $V_k$ to denote the expected cost of the optimal online algorithm when $k$ variables remain, and we use dynamic programming to obtain its lower bound. With the explicit form of the survival function $G(x)$, the recursion can be written in a closed form. 
For ease of notation, we let $\beta = 1-1/n$.
Then we have
\begin{equation}
    V_k =1 + \int_{1}^{V_{k-1}} x^{-\frac{1}{n}} \,\d x = 1 + \frac{1}{1-\frac{1}{n}} \cdot \left( V_{k-1}^{1-\frac{1}{n}}-1 \right) 
    = \frac{1}{\beta}\cdot V_{k-1}^{\beta} + \left( 1-\frac{1}{\beta} \right).
    \label{equation:recursion_V_k_lb}
\end{equation}

Using a similar approach as in the previous subsection, in the following, we prove that $V_k \geq b^{\beta^k}$ for all $k\geq 0$ by induction.
\begin{itemize}
    \item 
    Base case: By definition, $V_0 = b$ satisfies $V_0 \geq b^{\beta^0} = b$.
    \item 
    Assume that $V_{k-1} \geq b^{\beta^{k-1}}$ holds. For $V_k$, we have
    \begin{align*}
        V_k &= \frac{1}{\beta}\cdot V_{k-1}^{\beta} - (\frac{1}{\beta}-1)
        \geq \frac{1}{\beta}\cdot (b^{\beta^{k-1}})^{\beta} - (\frac{1}{\beta}-1) \\
        &= \frac{1}{\beta}\cdot b^{\beta^k}- (\frac{1}{\beta}-1) - b^{\beta^k}+b^{\beta^k} 
        =(\frac{1}{\beta}-1)\cdot (b^{\beta^k}-1) + b^{\beta^k}
        \geq b^{\beta^k},
    \end{align*}
    where the last inequality holds by $b\geq 1$ and $\beta < 1$.
\end{itemize}

Now we have $V_n \geq b^{\beta^n}$, which implies a lower bound on the competitive ratio:s
\begin{equation*}
    \frac{\ALG}{\OPT} \geq \frac{b^{(1-\frac{1}{n})^n}}{1+\ln b}. \qedhere
\end{equation*}
\end{proof}

\paragraph{Remark.}
The above closed-form lower bound is obtained via some lossy lower bounds on $\ALG$. Consequently, for small values of $b$, e.g., $b=n=2$, the lower bound $\frac{b^{(1-\frac{1}{n})^n}}{1+\ln b}$ could be smaller than $1$.
For specific values of $b$ and $n$, we can use the recursion in Equation~\eqref{equation:recursion_V_k_lb} to give a lossless characterization of $\ALG$.
Unfortunately, such recursion does not admit a closed-form expression.

\section{Distribution-Oblivious Online Algorithms}
\label{sec:oblivious-distributions}

In this section, we consider the case where the algorithm knows $b$ and $n$, but does not know the exact CDF of the $b$-bounded distribution. 
We study both single-threshold and multi-threshold algorithms, and derive near-optimal competitive ratios.

\subsection{Single-Threshold Algorithm}
\label{sec:single-threshold}

As introduced, a distribution-oblivious single-threshold algorithm fixes a threshold $\theta$ that depends only on $b$ and $n$, and accepts the first variable whose realization is below $\theta$. If none of the first $n-1$ variables are accepted, the last variable is accepted regardless of its realization.
Single-threshold algorithms form a practical class of strategies in many real-world applications. For instance, in the air ticket booking problem, a natural approach is to purchase a ticket whenever its price falls below a predetermined threshold. Hence, it is of practical interest to determine how to set such a threshold optimally and what competitive ratio it guarantees.

Interestingly, we show that the optimal choice is $\theta = \sqrt{b}$, which is independent of $n$. The resulting single-threshold algorithm achieves a competitive ratio of $O(b^{\frac{n-1}{2n}})$. Moreover, we prove that no single-threshold algorithm can attain a ratio of $o(b^{\frac{n-1}{2n}})$, even when the distribution is known, indicating that our bound is asymptotically optimal.

Notably, when $n = 2$, the algorithm achieves a $b^{1/4}$-competitive ratio, which matches that of the optimal online algorithm for known distributions. This observation suggests that distribution-oblivious algorithms, with appropriately chosen thresholds, can be as competitive as distribution-aware algorithms. 
On the other hand, the hardness result of $\Omega(b^{\frac{n-1}{2n}})$ implies that single-threshold strategies are sub-optimal for large $n$.
This insight further motivates our study of distribution-oblivious algorithms with multiple thresholds.

We summarize the algorithm and its competitive guarantees as follows.

\begin{tcolorbox}[title=Single-threshold Algorithm, label=algorithm:single-threshold]
    \textbf{Set threshold $\theta = \sqrt{b}$.}
    \textbf{Upon the arrival of $X_i$, where $i=1,\ldots,n$:}
    \begin{itemize}
        \item
        for $i\leq n-1$: if $X_i \leq \theta$, accept $X_i$ and terminate;
        \item
        for $i=n$: accept $X_n$ and terminate.
    \end{itemize}
\end{tcolorbox}

\begin{theorem}
\label{thm:ub-single-threshold}
    For cost-minimization IID prophet inequality with an unknown $b$-bounded distribution, there exists a single-threshold algorithm that achieves a competitive ratio of $2\cdot b^{\frac{n-1}{2n}} - b^{\frac{1}{2n}}$.
\end{theorem}
\begin{proof}
Given that $\theta$ is used to decide the acceptance of the first $n-1$ variables, if the algorithm accepts a variable before the last round, its expected cost is $\E[X\mid X\leq\theta]$; otherwise, its expected cost is $\E[X]$. Thus, we have
\begin{align*}
    \ALG &= \sum_{i=1}^{n-1} F(\theta) \cdot G(\theta)^{i-1}\cdot \E[X\mid X\leq\theta] + G(\theta)^{n-1}\cdot \E[X] \\
    &= \left(\sum_{i=1}^{n-1} G(\theta)^{i-1}\right) \cdot \left(\int_{0}^{\theta}G(x) \,\d x - \theta \cdot G(x) \right) + G(\theta)^{n-1} \cdot \E[X] \\
    &= \left(\int_{0}^{\theta} G(x) \,\d x \cdot \sum_{i=1}^{n-1} G(\theta)^{i-1}\right) - \left(\theta \cdot \sum_{i=1}^{n-1}G(\theta)^i\right) + G(\theta)^{n-1} \cdot \E[X] \\
    &=\int_{0}^{\theta} G(x) \,\d x + \left(\left(\int_{0}^{\theta} G(x) \,\d x - \theta\right) \cdot \sum_{i=2}^{n-1} G(\theta)^{i-1}\right) - \theta \cdot G(\theta)^{n-1} + G(\theta)^{n-1} \cdot \E[X] \\
    &\leq \int_{0}^{\theta} G(x) \,\d x  - \theta \cdot G(\theta)^{n-1} + G(\theta)^{n-1} \cdot \E[X] \\
    &= \int_{0}^{\theta} G(x) \,\d x + G(\theta)^{n-1} \cdot \left(\int_{0}^{b} G(x) \,\d x - \theta \right),
\end{align*}
where the second equality holds by Lemma~\ref{lem:conditioned-expectation} and the first inequality holds by $\int_{0}^{\theta} G(x) \,\d x \leq \theta$.

Our next step is to get rid of the dependence on $G(x)$ in the expression and relate the upper bound to $\OPT$.
By applying Hölder’s inequality using Lemma~\ref{lemma:upper-bound-G-by-OPT}, we have
\begin{align*}
    \ALG &\leq \int_{0}^{\theta} G(x) \,\d x + G(\theta)^{n-1} \cdot \left(\int_{0}^{b} G(x) \,\d x - \theta \right) \\
    &\leq \OPT^{\frac{1}{n}}\cdot \theta^{\frac{n-1}{n}} + G(\theta)^{n-1} \cdot \left( \OPT^{\frac{1}{n}}\cdot b^{\frac{n-1}{n}} - \theta \right) \\
    &\leq \OPT^{\frac{1}{n}}\cdot \theta^{\frac{n-1}{n}} + \left( \frac{\OPT^{\frac{1}{n}}}{\theta^{\frac{1}{n}}} \right)^{n-1} \cdot \left( \OPT^{\frac{1}{n}}\cdot b^{\frac{n-1}{n}} - \theta \right) \\
    & = \OPT\cdot \left(\frac{b}{\theta}\right)^{\frac{n-1}{n}} + \left( \OPT^{\frac{1}{n}}\cdot \theta^{\frac{n-1}{n}} - \OPT^{\frac{n-1}{n}}\cdot \theta^{\frac{1}{n}} \right),
\end{align*}
where the second inequality holds by Lemma~\ref{lemma:upper-bound-G-integral} and the third inequality holds by Lemma~\ref{lemma:upper-bound-G-by-OPT} (note that $\OPT^{\frac{1}{n}}\cdot b^{\frac{n-1}{n}} \geq b^{\frac{n-1}{n}} \geq \sqrt{b}$ for $n\geq 2$).
Furthermore, notice that
\begin{equation*}
    \OPT^{\frac{1}{n}}\cdot \theta^{\frac{n-1}{n}} - \OPT^{\frac{n-1}{n}}\cdot \theta^{\frac{1}{n}}
    \leq \OPT^{\frac{1}{n}}\cdot \theta^{\frac{n-1}{n}} - \OPT^{\frac{1}{n}}\cdot \theta^{\frac{1}{n}}
    \leq \OPT\cdot \left(\theta^{\frac{n-1}{n}} - \theta^{\frac{1}{n}}\right).
\end{equation*}

Plugging in this upper bound and substituting $\theta = \sqrt{b}$, we obtain
\begin{align*}
    \ALG \leq \OPT \cdot \left(\theta^{\frac{n-1}{n}} + b^{\frac{n-1}{n}}\cdot \theta^{-\frac{n-1}{n}} - \theta^{\frac{1}{n}}  \right) 
    = \OPT \cdot \left(2\cdot b^{\frac{n-1}{2n}} - b^{\frac{1}{2n}}\right),
\end{align*}

Rearranging this inequality completes the proof.
\end{proof}

\subsection{Optimality of Single-Threshold Algorithm}

In this subsection, for every $n\geq 2$, we provide a hard instance and prove that the competitive ratio of any single-threshold algorithm is at least $\Omega(b^{\frac{n-1}{2n}})$, even if the algorithm knows the distribution.
In other words, the hardness comes from the limited flexibility to set the thresholds, not because of the unawareness of the distribution.

\begin{theorem}
\label{thm:lb-single-threshold}
    For cost-minimization IID prophet inequality with a known $b$-bounded distribution, every single-threshold algorithm has a competitive ratio of at least $\frac{1-b^{-1}}{6} \cdot b^{\frac{n-1}{2n}}$.
\end{theorem}

In what follows, we present a hard instance for a fixed $n$.
On a high level, we intend to construct an instance such that $\OPT$ is a constant, while the single-threshold algorithm is forced to set $\theta = \sqrt{b}$ and suffer a cost of $\Omega(b^{\frac{n-1}{2n}})$.
Consider a distribution with:
\begin{align*}
    X=
    \begin{cases}
        1 \quad & \text{with probability } 1-b^{-\frac{1}{2n}}, \\
        \sqrt{b} &  \text{with probability } b^{-\frac{1}{2n}} - b^{-\frac{1}{n}}, \\
        b \quad &  \text{with probability } b^{-\frac{1}{n}}.
    \end{cases}
\end{align*}

From this distribution, we have $G(1) = b^{-\frac{1}{2n}}$ and $G(\sqrt{b}) = b^{-\frac{1}{n}}$ (recall that $G(x) = \Pr [X > x]$). 
In the following, we prove that the competitive ratio of any single-threshold algorithm (with threshold being $\theta$) is at least $\frac{1-b^{-1}}{6} \cdot b^{\frac{n-1}{2n}}$.

For the offline optimum, we have
\begin{align*}
    \OPT \leq 1 + \sqrt{b} \cdot G(1)^n + b\cdot G(\sqrt{b})^n 
    =1 + \sqrt{b}\cdot b^{-\frac{1}{2}} + b \cdot b^{-1}
    =3.
\end{align*}

If the threshold $\theta$ is set to $b$, the algorithm is trivially bad-behaved, because $\E[X] \geq b\cdot b^{-\frac{1}{n}} = b^{\frac{n-1}{n}}$. 
Therefore, we assume that $\theta < b$.
For convenience, we use $T$ to denote the event that the algorithm receives a variable whose value is at most $\theta$ and $\bar{T}$ to denote that $T$ does not happen. 
Note that it is possible that the algorithm accepts the last variable $X_n$ and it happens to be at most $\theta$.
We can equivalently interpret the algorithm as using $\theta$ to decide \emph{all} variables (including $X_n$), but are forced to select $X_n$ if no such event happens.
Such an interpretation will also be used in the next section.
Therefore, we can express the expected cost of the algorithm as 
\begin{align*}
    \ALG = \Pr[T] \cdot \E[X\mid X\leq \theta] + \Pr[\bar{T}] \cdot \E[X\mid X> \theta].
\end{align*}

We consider two cases:
\begin{itemize}
    \item 
    Case 1: $\theta<\sqrt{b}$. In this case, the algorithm accepts only variables whose values are $1$ in the first $n-1$ steps, and we have
    \begin{align*}
        \ALG &= \left(1-G(1)^n\right) \cdot 1 + G(1)^n \cdot \E[X\mid X>1] \\
        &= 1- G(1)^n + G(1)^{n-1} \cdot \left(\sqrt{b} \cdot \left(b^{-\frac{1}{2n}} - b^{-\frac{1}{n}}\right) + b \cdot b^{-\frac{1}{n}}\right) \\
        &\geq 1-b^{-\frac{1}{2}} + b^{-\frac{n-1}{2n}} \cdot b \cdot b^{-\frac{1}{n}} \\
        &=1-b^{-\frac{1}{2}} + b^{\frac{n-1}{2n}}
        \geq b^{\frac{n-1}{2n}}.
    \end{align*}
    \item 
    Case 2: $\theta \geq \sqrt{b}$. Intuitively, due to the large threshold, event $T$ is very likely to happen, and when it happens, the algorithm has a high probability of receiving $\sqrt{b}$ instead of $1$. Specifically, we have
    \begin{align*}
        \ALG &\geq \left(1-G(\sqrt{b})^n\right) \cdot \E[X\mid X\leq \sqrt{b}] \\
        &= \frac{1-G(\sqrt{b})^n}{1-G(\sqrt{b})} \cdot \left( \left(1-G(1)\right)\cdot 1 + \left(G(1)-G(\sqrt{b})\right) \cdot \sqrt{b}\right) \\
        &=\frac{1-b^{-1}}{1-b^{-\frac{1}{n}}} \cdot \left( \left(1-b^{-\frac{1}{2n}}\right)\cdot 1 + \left(b^{-\frac{1}{2n}}-b^{-\frac{1}{n}}\right) \cdot \sqrt{b}\right)\\
        &=\frac{1-b^{-1}}{1-b^{-\frac{1}{n}}} \cdot\left(1-b^{-\frac{1}{2n}}\right) \cdot \left(1 + \sqrt{b} \cdot b^{-\frac{1}{2n}}\right) \\
        &= \frac{1-b^{-1}}{1+b^{-\frac{1}{2n}}} \cdot \left(1+b^{\frac{n-1}{2n}}\right) \geq \frac{1-b^{-1}}{2} \cdot b^{\frac{n-1}{2n}}.
    \end{align*}

    Since $b>1$ is a known constant for the algorithm, this lower bound is $\Omega(b^{\frac{n-1}{2n}})$.
\end{itemize}

Since the expected cost of the offline optimum is at most $3$, the competitive ratio of any single-threshold algorithm is at least $\frac{1-b^{-1}}{6} \cdot b^{\frac{n-1}{2n}}$, thereby proving Theorem~\ref{thm:lb-single-threshold}.

Recall that we also provided a hard instance (with $G(x) = x^{-1/n}$) in Section~\ref{ssec:lb-online-optimal} for the optimal online algorithm.
In fact, it can be shown that single-threshold algorithms also perform badly on the instance. 
However, the derived lower bound on the competitive ratio is slightly weaker than what we have shown above (see Appendix~\ref{sec:optimality-better}).

\subsection{Multi-threshold Algorithm}
\label{ssec:MTA}

As we have shown above, due to their limited flexibility, single-threshold algorithms fail to match the ratio of the optimal online algorithm for $n \geq 3$. This limitation naturally motivates the study of the more general class of multi-threshold algorithms, where a key challenge is to determine suitable (increasing) thresholds for the distribution-oblivious algorithm.

Recall that in Section~\ref{ssec:ub-online-optimal}, we proved that the expected cost of the optimal online algorithm with $k$ variables satisfies
$V_k \leq \OPT^{1-(1-1/n)^{k}} \cdot b^{(1-1/n)^k}$.
On the other hand, in Section~\ref{ssec:lb-online-optimal}, we provided a hard instance for which we have $V_k \geq b^{(1-1/n)^k}$.
In the natural regime where $\OPT = O(1)$, these bounds suggest that $V_k = \Theta(b^{(1-1/n)^k})$, which is independent of the underlying distribution. Furthermore, when $n = 2$, this characterization is consistent with the choice of $\theta = \sqrt{b}$ for the single-threshold algorithm.
Motivated by this evidence, we propose a multi-threshold distribution-oblivious algorithm whose competitive ratio matches the lower bound $b^{(1-1/n)^n}$ we derived for the optimal online algorithm.
The algorithm fixes a threshold $\theta_i$ to each variable $X_i$, where the thresholds are defined by
\begin{equation*}
    \theta_i = b^{\left(1-\frac{1}{n}\right)^{n-i}}, \quad \forall i\in \{1,2,\ldots,n\}.
\end{equation*}
By construction, the thresholds satisfy $\theta_1 \leq \theta_2 \leq \cdots \leq \theta_n = b$.

\begin{tcolorbox}[title=Multi-threshold Algorithm, label=algorithm:MTA]
    \textbf{Set threshold $\theta_i = b^{(1-1/n)^{n-i}}$ for every $i \in \{1,2,\ldots,n\}$.} 

    \smallskip
    
    \textbf{Upon the arrival of $X_i$, where $i=1,\ldots,n$:}
    \begin{itemize}
        \item
        if $X_i \leq \theta_i$, accept $X_i$ and terminate.
    \end{itemize}
\end{tcolorbox}

We first provide a characterization for the online algorithm and derive an upper bound on its expected cost.
Similar to our previous analysis for the single-threshold algorithms, we have
\begin{align*}
    \ALG & = \sum_{t=1}^n \left( \Pr[X_t \text{ is selected}]\cdot \E[X_t \mid X_t \leq \theta_t] \right) \\
    & = \sum_{t=1}^n \left( \left(\prod_{i=1}^{t-1} G(\theta_i)\right) \cdot F(\theta_t)\cdot \E[X \mid X \leq \theta_t] \right) \\
    & = \sum_{t=1}^{n} \left( \left( \prod_{i=1}^{t-1} G(\theta_i) \right) \cdot \left( \int_{0}^{\theta_t} G(s) \,\d s - \theta_t \cdot G(\theta_t)  \right) \right) \\
    &= \sum_{t=1}^{n} \left( \left( \prod_{i=1}^{t-1} G(\theta_i) \right) \cdot  \int_{0}^{\theta_t} G(s) \,\d s  \right) - \sum_{i=1}^{n} \left( \left( \prod_{i=1}^{t} G(\theta_i) \right) \cdot \theta_t\right) \\
    &= \int_{0}^{\theta_1}G(s) \,\d s + \sum_{t=2}^{n} \left( \left( \prod_{i=1}^{t-1} G(\theta_i) \right) \cdot  \left(\int_{0}^{\theta_t} G(s) \,\d s  -\theta_{t-1} \right) \right) - \left( \prod_{i=1}^{n} G(\theta_i) \right)\cdot \theta_n \\
    &\leq \int_{0}^{\theta_1}G(s) \,\d s + \sum_{t=2}^{n} \left( \left(\prod_{i=1}^{t-1} G(\theta_i) \right)\left(\int_{0}^{\theta_t} G(s) \,\d s -\theta_{t-1} \right) \right) 
\end{align*}
where the second equality holds since $X_t$ is selected if and only if all variables $X_1,\ldots,X_{t-1}$ exceed the corresponding thresholds and $X_t \leq \theta_t$, and the third equality holds by Lemma~\ref{lem:conditioned-expectation}.

\paragraph{High-level Idea.}
Intuitively, since our algorithm is distribution-oblivious, the survival function $G(\cdot)$ could take any form, and it is hard to give a straightforward upper bound.
On the other hand, we can fully control the thresholds $\theta_1,\ldots,\theta_n$.
Therefore, similar to the analysis for the single-threshold case, the high-level idea is to find a way to remove the dependence on the function $G(\cdot)$ in the upper bound for $\ALG$, i.e., upper bound $\ALG$ in terms of $\OPT$ and $\theta_1,\ldots,\theta_n$.
Then we can use the specific form of the thresholds to obtain an upper bound on the competitive ratio.

\medskip

Recall that in Section~\ref{sec:preliminaries}, we provide upper bounds for $G(\theta_i)$ and $\int_{0}^{\theta_i} G(s)\,\d s$ in terms of $\OPT$ and $\theta_i$. Using Lemma~\ref{lemma:upper-bound-G-by-OPT} and Lemma~\ref{lemma:upper-bound-G-integral}, we can reformulate the upper bound for $\ALG$ as
\begin{align*}
    \ALG & \leq \OPT^{1/n}\cdot \theta_1^{\frac{n-1}{n}} +  \sum_{t=2}^{n} \left( \left(\prod_{i=1}^{t-1} G(\theta_i) \right)\left(\OPT^{1/n} \cdot \theta_t^{\frac{n-1}{n}} -\theta_{t-1} \right) \right).
\end{align*}

Recall that $\theta_t = b^{(\frac{n-1}{n})^{n-t}}$.
Therefore we have $\theta_n = b$, $\theta_1 = b^{(\frac{n-1}{n})^{n-1}}$, and $\theta_{t-1} = \theta_t^{\frac{n-1}{n}}$ for all $t \in \{2,\ldots,n\}$. 
Consequently, for all $t \in \{2,\cdots,n\}$ we have
\begin{equation*}
    \OPT^{\frac{1}{n}} \cdot \theta_t^{\frac{n-1}{n}} - \theta_{t-1} \geq 0.
\end{equation*}

Thus, using Lemma~\ref{lemma:upper-bound-G-by-OPT}, the $\ALG$ can be upper bounded by
\begin{align*}
    \ALG & \leq \OPT^{\frac{1}{n}}\cdot \theta_1^{\frac{n-1}{n}} + \sum_{t=2}^{n}\left( \left(\prod_{i=1}^{t-1} \left( \frac{\OPT}{\theta_i}\right)^{\frac{1}{n}} \right) \cdot \left(\OPT^{\frac{1}{n}} \cdot \theta_t^{\frac{n-1}{n}} -\theta_{t-1} \right) \right) \\
    &= \OPT^{\frac{1}{n}} \cdot \theta_1^{\frac{n-1}{n}} + \sum_{t=2}^{n} \left( \frac{\OPT^{\frac{t-1}{n}}}{\left(\prod_{i=1}^{t-1} \theta_i\right)^{\frac{1}{n}}} \cdot \left(\OPT^{\frac{1}{n}} \cdot \theta_t^{\frac{n-1}{n}} -\theta_{t-1} \right)\right) \\
    &= \OPT^{\frac{1}{n}} \cdot \theta_1^{\frac{n-1}{n}} + \sum_{t=2}^{n} \left( \frac{\OPT^{\frac{t-1}{n}}\cdot \theta_{t-1}}{\left(\prod_{i=1}^{t-1} \theta_i\right)^{\frac{1}{n}}} \cdot \left(\OPT^{\frac{1}{n}} - 1 \right)\right)
\end{align*}
Next, we show a property regarding the threshold sequence.

\begin{lemma} \label{lemma:theta-t-over-prod=theta-1}
    For all $t\in \{2,3,\ldots,n\}$, we have
    \begin{equation*}
        \frac{\theta_{t-1}}{( \prod_{i=1}^{t-1} \theta_i )^{\frac{1}{n}}} = \theta_1^{\frac{n-1}{n}}.
    \end{equation*}
\end{lemma}
\begin{proof}
    For ease of notation, we consider the exponent of the LHS with base $b$:
    \begin{equation*}
        \ln_b\left( \frac{\theta_{t-1}}{\left( \prod_{i=1}^{t-1} \theta_i \right)^{\frac{1}{n}}}\right) = \left( \frac{n-1}{n} \right)^{n-t+1} - \frac{1}{n}\cdot \sum_{i=1}^{t-1} \left( \frac{n-1}{n} \right)^{n-i} = \left( \frac{n-1}{n} \right)^{n}.
    \end{equation*}

    Therefore, we have $\frac{\theta_{t-1}}{( \prod_{i=1}^{t-1} \theta_i )^{\frac{1}{n}}} = b^{(\frac{n-1}{n})^{n}} = \theta_1^{\frac{n-1}{n}}$.
\end{proof}

Using Lemma~\ref{lemma:theta-t-over-prod=theta-1}, we can further simplify the upper bound for $\ALG$:
\begin{align*}
    \ALG &\leq \OPT^{\frac{1}{n}} \cdot \theta_1^{\frac{n-1}{n}} + \sum_{t=2}^{n} \left(\OPT^{\frac{t-1}{n}}\cdot \theta_1^{\frac{n-1}{n}} \cdot \left(\OPT^{\frac{1}{n}}  -1 \right)\right) \\
    &= \theta_1^{\frac{n-1}{n}} \cdot \left(\OPT^{\frac{1}{n}} + \sum_{t=2}^{n} \left(\OPT^{\frac{t}{n}}  - \OPT^{\frac{t-1}{n}} \right)\right) \\
    &= \theta_1^{\frac{n-1}{n}} \cdot \OPT.
\end{align*}
Therefore, we have the following upper bound on the competitive ratio
\begin{equation*}
    \frac{\ALG}{\OPT} \leq \theta_1^{\frac{n-1}{n}} = b^{(1-\frac{1}{n})^n},
\end{equation*}
which matches the upper bound we derived in Section~\ref{ssec:ub-online-optimal} for the optimal online algorithm, and thus is near-optimal (up to a logarithmic factor). 

\paragraph{Extension to Limited Thresholds.} In Appendix~\ref{sec:k-threshold}, we show that if we are only allowed to use $k$ thresholds, say $k=O(1)$, then we can still obtain a competitive ratio of $O\left(\ln b \cdot b^{\left(1-1/(k+1)\right)^{k-1}}\right)$. 
This additional result shows that with a constant number of thresholds, we can approximate the optimal online algorithm very well.

\subsection{Numerical Results}
\label{ssec:numerical}

So far, we have established a nearly tight characterization of the competitive ratios for both single- and multi-threshold distribution-oblivious algorithms. While these ratios are polynomial in $b$, we show that their numerical values remain quite small due to the small exponent.
We summarize the competitive ratio for $n\in\{3,\infty\}$ and $b\in \{2,3,\ldots,10\}$ in Table~\ref{table:true-bounds-n=3} and Table~\ref{table:true-bounds-n=infty}.

\begin{table}[htbp]
\centering
\begin{minipage}{0.48\linewidth}
\centering
\caption{Competitive ratio guarantees for different $b$ when $n=3$.}
\label{table:true-bounds-n=3}
\begin{tabular}{ccc}
\toprule
\(b\) & Single-threshold & Multi-threshold \\
\midrule
2 & 1.398 & 1.228 \\
3 & 1.684 & 1.385 \\
4 & 1.915 & 1.508 \\
5 & 2.113 & 1.612 \\
6 & 2.287 & 1.701 \\
7 & 2.443 & 1.780 \\
8 & 2.586 & 1.852 \\
9 & 2.718 & 1.918 \\
10 & 2.842 & 1.979 \\
\bottomrule
\end{tabular}
\end{minipage}
\hfill
\begin{minipage}{0.48\linewidth}
\centering
\caption{Competitive ratio guarantees for different $b$ when $n=\infty$.}
\label{table:true-bounds-n=infty}
\begin{tabular}{ccc}
\toprule
\(b\) & Single-threshold & Multi-threshold \\
\midrule
2 & 1.829 & 1.291 \\
3 & 2.465 & 1.499 \\
4 & 3.000 & 1.666 \\
5 & 3.473 & 1.808 \\
6 & 3.899 & 1.934 \\
7 & 4.292 & 2.046 \\
8 & 4.657 & 2.149 \\
9 & 5.000 & 2.245 \\
10 & 5.325 & 2.333 \\
\bottomrule
\end{tabular}
\end{minipage}
\end{table}

\section{Non-IID Distributions}
\label{sec:non-IID}

In this section, we study the non-IID prophet inequality, where there are $n$ independent random variables $X_1,\ldots,X_n$, with variable $X_i$ drawn from distribution $D_i$. In the prophet inequality, the order of these random variables is adversarial. In the prophet secretary, the order of these random variables is uniformly random. We define $g_i = \Pr[X_i > \theta] = G_i(\theta)$ for $i=1\ldots,n$ where $G_i = 1 - F_i$ is the survival function of distribution $D_i$.
We find that the best possible competitive ratio for cost-minimization prophet inequality is $\Omega(\sqrt{b})$, even under the prophet secretary setting. Moreover, this ratio can be achieved using a distribution-oblivious single-threshold algorithm with $\theta = \sqrt{b}$.


We start with the non-IID prophet inequality. Interestingly, we show that there exists a single-threshold algorithm that is $2\sqrt{b}$-competitive, and is asymptotically optimal.

\begin{theorem}
\label{thm:adversary-order-single-threshold}
For cost-minimization prophet inequality with general (non-IID) $b$-bounded distributions, there exists a single-threshold algorithm that achieves a competitive ratio of $2\cdot\sqrt{b}$.
\end{theorem}
\begin{proof}
 We consider the single-threshold algorithm and use $\theta = \sqrt{b}$ as the threshold.  
As in the previous analysis for the IID case, if the algorithm accepts a variable below the threshold, the algorithm incurs a cost of at most $\theta$; otherwise, it incurs a cost of at most $b$. Thus, we have
\begin{align*}
    \ALG \leq \theta + b \cdot \prod_{i=1}^{n} g_i.
\end{align*}

For the offline optimum, we have
\begin{align*}
    \OPT &= 1 + \int_{1}^{b} \prod_{i=1}^{n} G_i(x) \,\d x 
    \geq 1+ \int_{1}^{\theta} \prod_{i=1}^{n} G_i(x) \,\d x \\
    &\geq 1+(\theta-1)\cdot \prod_{i=1}^{n} g_i
    \geq \max\left\{1, \theta\cdot \prod_{i=1}^{n} g_i \right\}.
\end{align*}

Then, we upper bound the competitive ratio as
\begin{align*}
    \frac{\ALG}{\OPT} \leq \frac{\theta +b\cdot \prod_{i=1}^{n} g_i}{\max \left\{1,\theta \cdot \prod_{i=1}^{n} g_i\right\}} 
    \leq \theta + \frac{b}{\theta} = 2\cdot \sqrt{b},
\end{align*}
where the last equality follows by $\theta = \sqrt{b}$.
\end{proof}

Next, we present an instance for which no (distribution-aware) algorithm can achieve a competitive ratio of $o(\sqrt{b})$. 
The instance is inspired by that of~\cite{esa/EsfandiariHLM15}, as we have mentioned in the introduction. 
Suppose there are two random variables $X_1$ and $X_2$ where 
\begin{align*}
    X_1 = \sqrt{b}, \quad 
    X_2 = \begin{cases}
    1 \quad \text{with probability } 1- \frac{1}{\sqrt{b}}, \\
    b \quad \text{with probability } \frac{1}{\sqrt{b}}.
\end{cases}
\end{align*}

Note that $\E[X_2] > \sqrt{b}$.
Therefore, regardless of which variable the algorithm selects, its expected cost is at least $\sqrt{b}$. 
However, the expected cost of the prophet is
\begin{equation*}
    \E[\min\{X_2, \sqrt{b}\}] = 1\cdot (1-1/\sqrt{b}) + \sqrt{b} \cdot 1/\sqrt{b}\leq 2,
\end{equation*}
which yields that no algorithm can achieve a competitive ratio better than ${\sqrt{b}}/{2}$. 

\paragraph{Prophet Secretary.}
It is natural to ask whether better competitive ratios can be obtained in the prophet secretary setting.
Unfortunately, we show that no algorithm can perform better than $\sqrt{b}/4$ in this setting.
For the same hard instance we presented above, since $X_1$ will be the first variable with probability $\frac{1}{2}$, the expected cost of the online algorithm is at least $\sqrt{b}/2$.
This yields a lower bound of $\sqrt{b}/4$ on the competitive ratio.
Note that we can generalize the instance to $n$ variables by adding $n-2$ dummy variables with $X_i = b$, which do not affect $\ALG$ or $\OPT$.

\section{Conclusion}

In this work, we study the cost-minimization prophet inequality under $b$-bounded distributions. We provide a nearly tight analysis of the optimal online algorithm and design near-optimal distribution-oblivious algorithms with single and multiple thresholds. Moreover, we extend our analysis to the non-IID setting and establish asymptotically optimal competitive ratios.
A natural direction for future work is to close the remaining $(1+\ln b)$ gap between the upper and lower bounds on the competitive ratio of the optimal online algorithm. It would also be interesting to identify other natural classes of distributions that admit competitive algorithms in the cost-minimization setting. Finally, understanding the optimal competitive ratio for the maximization version of the prophet inequality under $b$-bounded distributions remains an intriguing open problem.

\section*{Declaration on the Use of Generative AI}

All central research ideas, problem formulations, and mathematical results presented in this paper originated with the authors. Every theorem, proof, and technical claim included in the final manuscript was independently verified by the authors, who take full responsibility for the correctness, originality, and presentation of the work.

\bibliographystyle{alpha}
\bibliography{ref}

\newpage

\appendix
\section{A Continuous Perspective When \texorpdfstring{$n\to \infty$}{}}
\label{sec:continuous-analysis}

In this appendix, we provide an analysis under the continuous perspective for the case when $n\to \infty$, showing that the competitive ratio is $O(b^{1/e}\cdot \frac{\ln b}{\ln \ln b})$. In the continuous perspective, each variable arrives at time $t\in [0,1]$ chosen uniformly at random.
However, (for some technical reasons) we do not assume that the arrival times are independent.
Instead, we assume that exactly $t\cdot n$ variables arrive at or before time $t$.
This is equivalent to assuming random arrival order and setting the arrival time of the $i$-th agent as $t = i/n$.
An interesting observation regarding the multi-threshold algorithm we presented in Section~\ref{sec:oblivious-distributions} is that when $k \to \infty$, it becomes an online algorithm with time-dependent thresholds $\theta: [0,1] \to [1,b]$ defined as follows:
\begin{equation*}
    \theta(t) = b^{e^{t-1}}.
\end{equation*}

Notice that $\theta(t)$ is strictly increasing, and we have $\theta(0) = b^{1/e}$ and $\theta(1)=b$.
The online algorithm accepts a variable $X$ arriving at time $t$ if $X \leq \theta(t)$.
Note that the online algorithm is guaranteed to accept a variable, because $\theta(1) = b$.
Let $u(t)$ denote the probability that the algorithm has not accepted any variable by time $t$.
Note that $u(t)$ is decreasing in $t$, and we have $u(0) = 1$ and $u(1) = 0$.
Then, we write an upper bound on the cost of the algorithm:
\begin{align*}
    \ALG \leq \int_{0}^{1} (-u'(t)\cdot \theta(t)) \,\d t 
    = \theta(0) + \int_{0}^{1} u(t) \cdot \theta'(t) \,\d t,
\end{align*}
where the equality holds by integration by parts.

For $\OPT$, we have
\begin{align*}
    \OPT = \int_{0}^{b} G(x)^n \,\d x 
    \geq 1+ \int_{\theta(0)}^{\theta(1)} G (x)^n \,\d x 
    = 1 + \int_{0}^{1} G(\theta(t))^n \cdot \theta'(t) \,\d t.
\end{align*}

Therefore, we have
\begin{equation*}
    \frac{\ALG}{\OPT} \leq \theta(0) + \int_{0}^{1} \frac{u(t) \cdot \theta'(t)}{1 + \int_{0}^{1} G(\theta(t))^n \cdot \theta'(t) \,\d t} \,\d t.
\end{equation*}

\begin{lemma}\label{lemma:derivatives-of-ALG-over-OPT}
For all $t\in [0,1]$, we have
\begin{equation*}
    \frac{u(t) \cdot \theta'(t)}{1 + \int_{0}^{1} G(\theta(t))^n \cdot \theta'(t) \d t} \leq \theta(0)\cdot (\ln b)^{1-t}\cdot \sqrt{e}.
\end{equation*}    
\end{lemma}
\begin{proof}
    Fix any $t\in [0,1]$, the probability $u(t)$ that the algorithm has not accepted any variable is exactly the probability that all variables arrived before time $t$ have values larger than the thresholds.
    Therefore we have
    \begin{equation*}
        u(t) = \exp\left( n\cdot \int_0^t \ln\left( G(\theta(s)) \right) \,\d s \right).
    \end{equation*}
    
    On the other hand, we have
    \begin{align*}
        1 + \int_{0}^{1} G(\theta(t))^n \cdot \theta'(t) \,\d t
        & \geq \int_{0}^{t} G(\theta(s))^n \cdot \theta'(s) \,\d s
        + \int_{t}^{1} 1 \,\d s \\
        & \geq \exp\left( \int_0^t \ln\left( G(\theta(s))^n \cdot \theta'(s) \right) \,\d s \right) \\
        & = \exp\left( n\cdot \int_0^t \ln\left( G(\theta(s)) \right) \,\d s \right)\cdot \exp\left( \int_0^t \ln \theta'(s) \,\d s \right),
    \end{align*}
    where the second inequality holds by the (continuous) AM-GM inequality.
    Therefore we have
    \begin{equation*}
        \frac{u(t) \cdot \theta'(t)}{1 + \int_{0}^{1} G(\theta(t))^n \cdot \theta'(t) \,\d t} \leq \frac{\theta'(t)}{\exp\left( \int_0^t \ln \theta'(s) \,\d s \right)}.
    \end{equation*}

    Hence, it remains to show that
    \begin{equation*}
        \exp\left( \int_0^t \ln \theta'(s) \,\d s \right) \geq 
        \frac{\theta'(t)}{\theta(0)}\cdot \frac{1}{(\ln b)^{1-t}\cdot \sqrt{e}}.
    \end{equation*}

    By definition of $\theta$, we have
    \begin{equation*}
        \theta'(s) = b^{e^{s-1}}\cdot e^{s-1}\cdot \ln b,
    \end{equation*}
    and
    \begin{equation*}
        \int_0^t \ln \theta'(s) \,\d s
        = \left( e^{t-1} - \frac{1}{e} \right)\cdot \ln b + \frac{t^2}{2} - t + t\cdot \ln\ln b.
    \end{equation*}

    Then we have
    \begin{align*}
        \exp\left( \int_0^t \ln \theta'(s) \,\d s \right)
        & = \frac{b^{e^{t-1}}}{\theta(0)} \cdot e^{\frac{t^2}{2}-t} \cdot (\ln b)^t \\
        & = \frac{\theta'(t)}{\theta(0)} \cdot e^{\frac{t^2}{2}-2t+1} \cdot (\ln b)^{t-1} \\
        & \geq \frac{\theta'(t)}{\theta(0)}\cdot \frac{1}{(\ln b)^{1-t}\cdot \sqrt{e}}. \qedhere
    \end{align*}
\end{proof}

Given Lemma~\ref{lemma:derivatives-of-ALG-over-OPT}, we have
\begin{align*}
    \frac{\ALG}{\OPT} & \leq \theta(0) + \sqrt{e}\cdot \theta(0)\cdot \int_0^1 (\ln b)^{1-t} \,\d t \\
    & \leq \left( 1+\sqrt{e}\cdot \frac{\ln b}{\ln\ln b} \right)\cdot \theta(0) \\
    & = \left( 1+\sqrt{e}\cdot \frac{\ln b}{\ln\ln b} \right)\cdot b^{1/e}.
\end{align*}

\section{Another Hard Instance for Single-Threshold Algorithms}
\label{sec:optimality-better}

We analyze the competitive ratio of the single-threshold algorithm with threshold $\theta$ on the hard instance shown in Section~\ref{ssec:lb-online-optimal} with survival function
\begin{align*}
    G(x) = 
    \begin{cases}
        x^{-\frac{1}{n}} \quad & 1\leq x < b,\\
        0 &x\geq b,
    \end{cases}
\end{align*}

We express $\ALG$ as 
\begin{align*}
    \ALG & = \sum_{i=1}^{n-1} F(\theta) \cdot G(\theta)^{i-1} \cdot \E[X\mid X\leq \theta] + G(\theta)^{n-1} \cdot \E[X] \\
    & = \left( \int_{0}^{\theta} G(x) \,\d x - \theta \cdot G(\theta) \right) \cdot \sum_{i=1}^{n-1} G(\theta)^{i-1} + G(\theta)^{n-1} \cdot \int_{0}^{b} G(x) \,\d x.
\end{align*}

By definition of $G(x)$, we obtain
\begin{align*}
    \ALG &= \left(\int_{0}^{\theta} x^{-\frac{1}{n}}\,\d x - \theta \cdot \theta^{-\frac{1}{n}}\right) \cdot \sum_{i=1}^{n-1} \theta^{-\frac{i-1}{n}} + \theta^{-\frac{n-1}{n}} \cdot \int_{0}^{b} x^{-\frac{1}{n}} \,\d x \\
    &=\left( \frac{1}{1-\frac{1}{n}} \cdot \theta^{\frac{n-1}{n}} - \theta^{\frac{n-1}{n}} \right) \cdot \sum_{i=1}^{n-1}\theta^{-\frac{i-1}{n}} + \theta^{-\frac{n-1}{n}} \cdot \frac{1}{1-\frac{1}{n}} \cdot b^{\frac{n-1}{n}} \\
    &= \frac{1}{n-1}\cdot \sum_{i=1}^{n-1} \theta^{\frac{n-i}{n}} + \frac{n}{n-1} \cdot \left( \frac{b}{\theta} \right)^{\frac{n-1}{n}} \\
    &= \frac{1}{n-1}\cdot \frac{\theta^{\frac{n-1}{n}} - 1}{1 - \theta^{-\frac{1}{n}}} + \frac{n}{n-1} \cdot \left( \frac{b}{\theta} \right)^{\frac{n-1}{n}}.
\end{align*}

Using $1+x\leq e^x$ for all $x$, we have
\begin{equation*}
    1 - \theta^{-\frac{1}{n}} \leq \frac{\ln \theta}{n}.
\end{equation*}

Therefore we have
\begin{equation*}
    \ALG \geq \frac{n}{n-1}\cdot \frac{\theta^{\frac{n-1}{n}} - 1}{\ln \theta} + \frac{n}{n-1} \cdot \left( \frac{b}{\theta} \right)^{\frac{n-1}{n}}.
\end{equation*}

If $\theta < \sqrt{b}$, the second term $(b/\theta)^{\frac{n-1}{n}}$ is at least $b^{\frac{n-1}{2n}}$.
If $\theta \geq \sqrt{b}$, then the first term is at least
\begin{align*}
    \frac{\theta^{\frac{n-1}{n}} - 1}{\ln \theta} \geq \frac{2}{\ln b}\cdot\left( b^{\frac{n-1}{2n}} - 1\right).
\end{align*}

As shown in Section~\ref{ssec:lb-online-optimal}, we have $\OPT = 1+\ln b$. Therefore, the competitive ratio is at least 
\begin{align*}
    \frac{b^{\frac{n-1}{2n}} - 1}{(1+\ln b)\cdot \ln b} = \Omega\left(\frac{b^{\frac{n-1}{2n}}}{\ln^2 b}\right).
\end{align*}

\section{Multi-Threshold Algorithm with Limited Thresholds}
\label{sec:k-threshold}

In this appendix, we present an algorithm with $k$ thresholds where $k$ divides $n$.

\begin{tcolorbox}[title=Multi-threshold Algorithm, label=algorithm:MTA-k]
    \textbf{Set threshold $\theta_i = b^{\left( \frac{k}{k+1} \right)^{k-i}}$ for every $i \in \{1,2,\ldots,k\}$.} 

    \smallskip
    
    \textbf{Upon the arrival of $X_i$, where $i=1,\ldots,n$:}
    \begin{itemize}
        \item
        If $X_i \leq \theta_{\lceil \frac{i\cdot k}{n} \rceil}$, accept $X_i$ and terminate.
    \end{itemize}
\end{tcolorbox}

The algorithm divides the entire sequence into $k$ phases of equal length. The threshold of the $i$-th phase is $\theta_i$. As in the previous analysis, we first characterize the online algorithm and derive an upper bound on its expected cost. 
By a similar analysis as in Section~\ref{ssec:MTA}, we have
\begin{align*}
    \ALG \leq \int_{0}^{\theta_1}G(s) \,\d s + \sum_{t=2}^{n} \left( \left(\prod_{i=1}^{t-1} G(\theta_{\lceil \frac{i\cdot k}{n} \rceil}) \right)\left(\int_{0}^{\theta_{\lceil \frac{i\cdot k}{n} \rceil}} G(s) \,\d s -\theta_{\lceil \frac{(i-1)\cdot k}{n} \rceil} \right) \right) 
\end{align*}
Note that if $\theta_{\lceil \frac{i\cdot k}{n} \rceil} = \theta_{\lceil \frac{(i-1)\cdot k}{n} \rceil}$, we have
\begin{align*}
    \int_{0}^{\theta_{\lceil \frac{i\cdot k}{n} \rceil}} G(s) \,\d s -\theta_{\lceil \frac{(i-1)\cdot k}{n} \rceil} \leq \lceil \frac{i\cdot k}{n} \rceil - \lceil \frac{(i-1)\cdot k}{n} \rceil = 0.
\end{align*}

Therefore, in each phase, only the first variable contributes a non-negative cost to $\ALG$. Thus, we drop all the negative terms and upper bound $\ALG$ by
\begin{align}
\label{eq:k-threshold-ALG}
    \ALG &\leq \int_{0}^{\theta_1} G(s) \,\d s + \sum_{z=2}^{k} \left(\left(\prod_{i=1}^{z-1} G(\theta_i)^{\frac{n}{k}} \right)\left( \int_{0}^{\theta_z} G(s) \,\d s - \theta_{z-1} \right)\right) \nonumber\\
    &\leq \int_{0}^{\theta_1} G(s) \,\d s + \sum_{z=2}^{k} \left(\left(\prod_{i=1}^{z-1} G(\theta_i)^{\frac{n}{k}} \right)\cdot \left( \theta_z - \theta_{z-1} \right)\right).
\end{align}

Using $\theta_z = b^{\left(\frac{k}{k+1}\right)^{k-z}}$ for $z \in \{1,2,\ldots,k\}$, we have
\begin{align*}
    \theta_z-\theta_{z-1} = b^{\left(\frac{k}{k+1}\right)^{k-z}} - b^{\left(\frac{k}{k+1}\right)^{k-z+1}} = \theta_z \cdot \left(1 - b^{\left(\frac{k}{k+1}\right)^{k-z+1} - \left(\frac{k}{k+1}\right)^{k-z}}\right) \leq \theta_z\cdot \frac{\ln b}{k+1},
\end{align*}
where the inequality holds by the standard inequality $e^{-x} \geq 1-x$.

Therefore, Inequality~\eqref{eq:k-threshold-ALG} becomes
\begin{align*}
    \ALG &\leq \int_{0}^{\theta_1} G(s) \,\d s + \sum_{z=2}^{k} \left( \left(\prod_{i=1}^{z-1} G(\theta_i)^{\frac{n}{k}}\right) \cdot \theta_z \cdot \frac{\ln b}{k+1} \right) \\
    &\leq \OPT^{\frac{1}{n}} \cdot \theta_1^{\frac{n-1}{n}} + \sum_{z=2}^{k} \left(\theta_z \cdot \frac{\ln b}{k+1} \cdot \prod_{i=1}^{z-1}\left(\frac{\OPT}{\theta_i}\right)^{\frac{1}{k}}  \right) \\
    &= \OPT^{\frac{1}{n}}\cdot \theta_1^{\frac{n-1}{n}} + \sum_{z=2}^{k} \frac{\theta_1}{k+1} \cdot \ln b \cdot \OPT^{\frac{z-1}{k}} \\
    &\leq \OPT \cdot \left( \theta_1 + \frac{k-1}{k+1}\cdot \ln b \cdot \theta_1 \right),
\end{align*}
where the equality holds by $\frac{\theta_z}{\prod_{i=1}^{z-1} \theta_i^{\frac{1}{k}}} = \theta_1$ for all $z \in \{1,\ldots,k\}$.

Therefore, the competitive ratio is
\begin{equation*}
    \frac{\ALG}{\OPT} \leq \theta_1 + \frac{k-1}{k+1}\cdot \ln b \cdot \theta_1 = O(\ln b \cdot b^{\left(\frac{k}{k+1}\right)^{k-1}})
\end{equation*}

\end{document}